\documentclass[submission,copyright,noderivs]{eptcs}
\providecommand{\event}{AFL 2023} 
\usepackage{graphicx} 
\usepackage{mathrsfs}
\usepackage{amssymb}
\usepackage{amsmath}
\usepackage{amsthm}
\usepackage{pifont}
\usepackage[dvipsnames]{xcolor}
\usepackage{refcount}

\theoremstyle{plain}
\newtheorem{theorem}{Theorem}
\newtheorem{proposition}[theorem]{Proposition}
\newtheorem{lemma}[theorem]{Lemma}

\theoremstyle{definition}
\newtheorem{example}[theorem]{Example}

\newcommand{\lfam}{\mathscr{L}}

\newcommand{\dpwk}{\textrm{DPWK}}
\newcommand{\dpwkm}[1]{\dpwk(#1)}
\newcommand{\revpwk}{\textrm{REV-PWK}}
\newcommand{\revpwkm}[1]{\revpwk(#1)}
\newcommand{\invdelta}{\delta^{\scriptscriptstyle\leftarrow}}

\newcommand{\invmu}{\mu^{\scriptscriptstyle\leftarrow}}
\newcommand{\invvdash}{\mathrel{\vdash^{\protect\raisebox{2pt}{$\scriptscriptstyle\leftarrow$}}}}

\newcommand{\dfa}{\textrm{DFA}}
\newcommand{\lba}{\textrm{LBA}}
\newcommand{\ora}{\overrightarrow}
\newcommand{\ula}{\underleftarrow}

\newcommand{\valc}{\textrm{VALC}}

\newcommand{\leftend}{\mathord{\vartriangleright}}
\newcommand{\rightend}{\mathord{\vartriangleleft}}

\newcommand{\dollar}{\texttt{\$}}

\newcommand{\subtext}[1]{\textnormal{\scriptsize #1}}
\newcommand{\cyes}{\textcolor{PineGreen}{\ding{51}}}%
\newcommand{\cno}{\textcolor{red}{\ding{55}}}%

\newcommand{\eoe}{\ifmmode$\hspace*{\fill}$\blacksquare\else\hspace*{\fill}$\blacksquare$\fi\smallskip}
\title{Reversible Two-Party Computations}

\author{Martin Kutrib and Andreas Malcher
\institute{%
  Institut f\"ur Informatik, Universit\"at Giessen\\
  Arndtstr.~2, 35392 Giessen, Germany}
\email{$\{$kutrib,andreas.malcher$\}$@informatik.uni-giessen.de}
}
\def\titlerunning{Reversible Two-Party Computations}
\def\authorrunning{M.~Kutrib, A.~Malcher}

\begin{document}

\maketitle

\begin{abstract}
Deterministic synchronous systems consisting of two finite automata
running in opposite directions on a shared read-only input are
studied with respect to their ability 
to perform reversible computations, which means that the automata
are also backward deterministic and, thus, are able to uniquely
step the computation back and forth.
We study the computational capacity of such devices and obtain
on the one hand that
there are regular languages that cannot be accepted by
such systems. On the other hand, such systems can accept even
non-semilinear languages. 
Since the systems communicate by sending messages, we consider
also systems where the number of messages sent during a
computation is restricted. We obtain a finite hierarchy with
respect to the allowed amount of communication inside the
reversible classes and separations to general, not
necessarily reversible, classes. Finally, we study closure
properties and decidability questions and obtain that the
questions of emptiness, finiteness, inclusion, and equivalence
are not semidecidable if a superlogarithmic
amount of communication is allowed.
\end{abstract}

\section{Introduction}

Watson-Crick automata have been introduced 
in~\cite{Freund:1997:wcfa:proc} as a formal model for 
DNA computing. The motivation for such automata comes 
from processes observed in nature and laboratories. 
Basically, the idea is to consider finite automata with 
two reading heads that run on either strand of a double 
stranded DNA-molecule. It is noted 
in~\cite{nagy:2007:ofpt3pswcfa:proc} that in nature enzymes 
moving along DNA strands may obey the biochemical direction 
of the single strands of the DNA sequence.
Hence, so-called $5'\to 3'$ Watson-Crick automata
have been introduced in~\cite{nagy:2007:ofpt3pswcfa:proc},
which are two-head finite automata where the heads start at
opposite ends of a strand and move in opposite physical directions.
It is known that no additional information is encoded in
the second strand provided that the complementarity relation 
of the double stranded sequence is one-to-one. In this case,
$5'\to 3'$ Watson-Crick automata share a common input sequence.

Watson-Crick automata and $5'\to 3'$ Watson-Crick automata
have intensively been investigated in the last years from
different points of view. Descriptional complexity aspects
of Watson-Crick automata are studied in~\cite{czeizler:2009:odcwca}.
$5'\to 3'$ Watson-Crick automata with several runs, which means
that both heads are sweeping between both ends of the input,
are investigated in~\cite{leupold:2010:ftwcsr} and a hierarchy 
with respect to the number of runs has been obtained. 
The aspect of the amount of communication between the two heads
that is necessary in accepting computations is highlighted
in~\cite{kutrib:2010:tpwcc} where $5'\to 3'$ Watson-Crick automata
with restricted communication are introduced and a finite
hierarchy concerning the amount of communication could be obtained.
The concept of sensing heads, where one head can sense the
presence of the other head, has been applied to 
$5'\to 3'$ Watson-Crick automata in~\cite{nagy:2013:h5s3sswcfal,nagy:2017:answkac}.
The concept of jumping automata, where the input is processed
in a discontinuous way, has been introduced and investigated
for $5'\to 3'$ Watson-Crick automata in~\cite{kocman:2022:ajwkfam}.
Finally, the impact of replacing
the underlying devices of finite automata by finite transducers
or pushdown automata is studied in~\cite{nagy:2021:od1lswkfst} 
and in~\cite{chatterjee:2017:wkpa,nagy:2020:wkpa}, respectively. 

Another line of research in recent years is the study of
reversible devices. Here, a computation is considered
\emph{reversible} if every configuration
has at most one unique successor configuration and 
at most one unique predecessor configuration.
The study of such devices that perform logically reversible
computations is motivated by Landauer's question of whether
logical irreversibility is an unavoidable feature
of useful computers. This question is of particular interest,
since Landauer has demonstrated that whenever a physical computer
throws away information about its previous state it must generate 
a corresponding amount of entropy that results in heat dissipation.
A detailed discussion and suitable references can be found
in~\cite{Bennet:1973:lrc}. 
Reversible variants of many computational models have been
studied in the literature. 
For Turing machines the first investigations on reversible
computations date back to the sixties of the last century.
It is shown in the work of Lecerf~\cite{lecerf:lmmtr:1963}
and Bennett~\cite{Bennet:1973:lrc} that it is possible
for every Turing machine to construct an equivalent 
reversible Turing machine. Hence, every irreversible
computation can be made reversible. This is no longer true
if finite automata are considered. 
On the one hand, it is known that reversible one-way deterministic finite automata
are computationally weaker than one-way deterministic finite automata in
general~\cite{Angluin:1982:irl} (cf. also~\cite{holzer:2018:mrdfa}). On the other hand, two-way 
deterministic finite automata and reversible two-way 
deterministic finite automata are equally powerful~\cite{kondacs:1997:pqfsa}. 
Similar results are known for multihead finite 
automata. In case of one-way motion, the reversible variant
is computationally weaker than the general model
(\cite{kutrib:2017:owrmhfa}), whereas
in case of two-way motion the computational power of the
reversible variant and the general model 
coincides~\cite{morita:2011:twrmhfa}.
Several more types of devices as, for example, queue
automata~\cite{kutrib:2016:rqa}, one-way counter machines
with multiple counters~\cite{kutrib:2022:rcoowca},
and parallel communicating finite automata~\cite{bordihn:2021:rpcfas}
have been investigated with respect to reversibility.
An overview of the topic is given in~\cite{kutrib:2014:arca}.

The aspect of reversibility has been studied for
Watson-Crick automata in~\cite{chatterjee:2017:rwka}.
One result is that every regular language can be accepted
by a reversible Watson-Crick automaton. Here, it is essential
that the complementarity relation 
of the double stranded sequence is \emph{not} one-to-one.
If the complementarity relation is one-to-one, another
result of~\cite{chatterjee:2017:rwka} gives that the computational 
power of reversible Watson-Crick automata and reversible
two-head finite automata (\cite{kutrib:2017:owrmhfa}) coincides.
In this paper, we study $5'\to 3'$ Watson-Crick automata
having a one-to-one complementarity relation and to
differentiate the notation from other variants we will
call the devices in question \emph{two-party Watson-Crick systems}.
This paper can be seen as a continuation of~\cite{kutrib:2010:tpwcc}
where communication restricted two-party Watson-Crick systems
are introduced and a strict four-level hierarchy depending on the
number of messages sent was established, where the levels are given
by $O(1)$, $O(\log(n))$, $O(\sqrt{n})$, and $O(n)$ messages allowed.
Moreover, it could be shown that the questions of emptiness,
finiteness, inclusion, and equivalence are not semidecidable, 
that is, not recursively enumerable, even if the communication 
is reduced to a limit $O(\log(n)\cdot\log(\log(n)))$. 
Here, we complement these results. 
After defining the model and giving two illustrating 
examples in Section~\ref{sec:prelim} we show 
in Section~\ref{sec:rev-versus-irrev} that there are
regular languages which can not be accepted by any
reversible two-party Watson-Crick systems with any amount
of communication. This is in strong contrast to general
two-party Watson-Crick systems where no communication is
necessary to accept regular languages. This result can be
used in Section~\ref{sec:closures} in which closure properties are investigated.
It turns out that reversible two-party Watson-Crick systems 
are closed under complementation and reversal, whereas they
are not closed under union, intersection, intersection with
regular languages, concatenation, iteration, length-preserving
homomorphism, and inverse homomorphism. 
In Section~\ref{sec:restr-comm}, we can extend the 
strict four-level hierarchy depending on the
number of messages sent from~\cite{kutrib:2010:tpwcc} to
reversible two-party Watson-Crick systems. Moreover, we obtain
that for every level the reversible systems are computationally weaker than
the general systems. Finally, we discuss in Section~\ref{sec:deci}
decidability questions. In a first step, we show that the questions
of emptiness, finiteness, inclusion, and equivalence are not
semidecidable for reversible two-party Watson-Crick systems
essentially disregarding the number of messages communicated.
In a second step, we refine the argumentation and apply and
adapt a result from~\cite{kutrib:2017:owrmhfa} which enables
us to show that the questions
of emptiness, finiteness, inclusion, and equivalence are not
semidecidable for reversible two-party Watson-Crick systems
even if the number of messages allowed is bounded by
$O(\log(n)\cdot\log(\log(n)))$.

\section{Definitions and Preliminaries}\label{sec:prelim}

We denote the set of nonnegative integers by $\mathbb{N}$. 
We write $\Sigma^*$ for the set of all words over the finite alphabet $\Sigma$.
The empty word is denoted by~$\lambda$, and 
$\Sigma^+ = \Sigma^* \setminus \{\lambda\}$. The reversal of a word $w$ is
denoted by $w^R$ and for the length of~$w$ we write~$|w|$. 
We use $\subseteq$ for inclusions and~$\subset$ for strict inclusions. 

A two-party Watson-Crick system is a device of two finite automata 
working independently and in opposite directions on a common read-only input
data. The automata communicate by broadcasting messages. 
The transition function of a single automaton depends on its current state,
the currently scanned input symbol, and the message currently received from 
the other automaton. Both automata work synchronously and the
messages are delivered instantly.
Whenever the transition function of (at least) one of the 
single automata is undefined the whole systems halts. 
The input is accepted if at least one of the automata is in an 
accepting state. 
A formal definition is as follows.

\begin{sloppypar}
A \emph{deterministic two-party Watson-Crick system} $(\dpwk)$ 
is a construct 
\mbox{$\mathcal{A}=\langle \Sigma,M, \leftend, \rightend, A_1, A_2\rangle$,} 
where 
$\Sigma$ is the finite set of \emph{input symbols},~$M$ is the set of 
possible \emph{messages},
$\leftend\notin\Sigma$ and $\rightend\notin \Sigma$ 
are the \emph{left and right endmarkers}, and
each $A_i=\langle Q_i, \Sigma, \delta_i, \mu_i, q_{0,i}, F_i\rangle$, $i\in\{1,2\}$, 
is basically a \emph{deterministic finite automaton} with \emph{state set} $Q_i$, 
\emph{initial state} $q_{0,i}\in Q_i$, and set of
\emph{accepting states} $F_i\subseteq Q_i$. Additionally, each $A_i$ has a 
\emph{broadcast function} 
$\mu_i: Q_i\times (\Sigma\cup\{\leftend,\rightend\}) \to M\cup\{\bot\}$
which determines the message \emph{to be sent}, where 
$\bot\notin M$ means \emph{nothing to send}, and a
\emph{(partial) transition function} 
\mbox{$\delta_i: Q_i \times (\Sigma \cup \{\leftend,\rightend\})\times
(M\cup\{\bot\}) \to Q_i \times \{0,+\}$,}
where~$+$ means to move the head one
square and~$0$ means to keep the head on the current square.
\end{sloppypar}

The automata $A_1$ and $A_2$ are called \emph{components} of the
system $\mathcal{A}$, where the so-called \emph{upper} component~$A_1$
starts at the left end of the input and
moves from left to right, and the \emph{lower} component $A_2$ starts at the right 
end of the input and moves from right to left.
A \emph{configuration} of $\mathcal{A}$ is represented by a string 
$\leftend v_1 \ora{p} x v_2 y\ula{q} v_3 \rightend$,
where $v_1xv_2yv_3$ is the input and it is understood that
component $A_1$ is in state~$p$ with its head scanning symbol~$x$, 
and component $A_2$ is in state $q$ with its head scanning symbol $y$.
System $\mathcal{A}$ starts with component~$A_1$ in its initial state
scanning the left endmarker and component $A_2$ in its initial state scanning
the right endmarker. 
So, for input $w\in \Sigma^*$, the initial configuration is
$\ora{q_{0,1}}\leftend w \rightend\ula{q_{0,2}}$.
A computation of $\mathcal{A}$ is a sequence of configurations
beginning with an initial configuration.
One step from a configuration to its successor
configuration is denoted by~$\vdash$.
Let $w= a_1 a_2\cdots a_n$ be the input, $a_0= \leftend$, and $a_{n+1}=
\rightend$, then we set
$$
a_0 \cdots a_{i-1} \ora{p}a_i \cdots a_j \ula{q} a_{j+1}\cdots a_{n+1}
\vdash
a_0 \cdots a_{i'-1} \ora{p_1}a_{i'} \cdots a_{j'} \ula{q_1} a_{j'+1}\cdots a_{n+1}
,$$
for $0\leq i,j \leq n+1$,
iff
$\delta_1(p,a_i,\mu_2(q,a_j)) = (p_1,d_1)$ and 
$\delta_2(q,a_j,\mu_1(p,a_i)) = (q_1,d_2)$,
$i'=i+d_1$ and $j'=j-d_2$.
As usual we define the
reflexive, transitive closure of $\vdash$ by $\vdash^*$.  

A computation \emph{halts} when the successor configuration is 
not defined for the current configuration. This may happen
when the transition function of one component is not defined.
The language $L(\mathcal{A})$ accepted by a $\dpwk$ $\mathcal{A}$
is the set of inputs $w\in \Sigma^*$ such that there is 
some computation beginning with the initial configuration for $w$ 
and halting with at least one component being in an accepting state.

Now we turn to \emph{reversible} two-party Watson-Crick systems.
Basically, reversibility is meant with respect to the possibility of 
stepping the computation back and forth. 
So, the system has also to be backward deterministic.
That is, any configuration must have at most one 
predecessor which, in addition, is computable by a two-party Watson-Crick system.
In particular for the read-only input tape, the 
machines reread the input symbol which they have been read in a preceding 
forward computation step. 
Therefore, for reverse computation steps 
the head of the upper component is either moved to the \emph{left} or 
stays stationary, whereas the head of the lower component is either moved to the
\emph{right} or stays stationary. One can imagine that in a forward step, 
first the input symbol is read and then the input head is moved to its new
position, whereas in a backward step, first the input head is moved to its new
position and then the input symbol is read.

So, a deterministic two-party Watson-Crick system $\mathcal{A}$ is said to be \emph{reversible} ($\revpwk$) if  
and only if there exist \emph{reverse transition functions}
\mbox{$\invdelta_i: Q_i \times (\Sigma \cup \{\leftend,\rightend\})\times
(M\cup\{\bot\}) \to Q_i \times \{0,-\}$} and
\emph{reverse broadcast functions}
\mbox{$\invmu_i: Q_i\times (\Sigma\cup\{\leftend,\rightend\}) \to M\cup\{\bot\}$}
inducing a relation~$\invvdash$ from a
configuration to its \emph{predecessor configuration}, such that  
$$a_0 \cdots a_{i'-1} \ora{p_1}a_{i'} \cdots a_{j'} \ula{q_1} a_{j'+1}\cdots
a_{n+1} \invvdash a_0 \cdots a_{i-1} \ora{p}a_i \cdots a_j \ula{q}
a_{j+1}\cdots a_{n+1}
$$
\centerline{if and only if}
$$a_0 \cdots a_{i-1} \ora{p}a_i \cdots a_j \ula{q} a_{j+1}\cdots a_{n+1}
\vdash
a_0 \cdots a_{i'-1} \ora{p_1}a_{i'} \cdots a_{j'} \ula{q_1} a_{j'+1}\cdots
a_{n+1}.
$$

In the following, we study the impact of communication in 
deterministic two-party Watson-Crick systems.
The communication is measured by the total number of messages sent 
during a computation, where it is understood that $\bot$
means no message and, thus, is not counted. 

Let $f:\mathbb{N}\to\mathbb{N}$ be a mapping. 
If all $w\in L(\mathcal{A})$
are accepted with computations where the total number of
messages sent is bounded by $f(|w|)$,
then $\mathcal{A}$ is said to be \emph{communication bounded by~$f$.}
We denote the class of $\dpwk$s that are communication bounded by 
$f$ by $\dpwkm{f}$. In case of reversible $\dpwk$s we have to
consider the number of messages sent in reverse computations as well.
If all $w\in L(\mathcal{A})$
are accepted with computations where the total number of
messages sent in forward computations and in reverse computations
is each bounded by $f(|w|)$,
then $\mathcal{A}$ is said to be \emph{communication bounded by~$f$}
and the corresponding class of $\revpwk$s is denoted by $\revpwkm{f}$.

In general, the \emph{family of languages accepted} by devices of type~$X$
is denoted by~$\lfam(X)$.
To illustrate the definitions we start with two examples.

\begin{example}\label{exa:anbn}
The non-regular language 
$L=\{\, a^n b^n \mid n\ge 1\,\}$ is 
accepted by a $\revpwk$.
The principal idea of the construction is that the upper component
starts with one time step delay and then moves its head with
maximum speed to the right, whereas the lower component immediately
starts to move its head with maximum speed to the left. Both
components communicate in every time step the symbol they read.
When the lower component has read the rightmost $a$ of the $a$-block
after having passed the $b$-block, the transition functions ensure
that the upper component has to read the leftmost $b$ of the 
$b$-block after having passed the $a$-block. When the lower
component has reached the left endmarker, it waits for one time
step. To accept the input, the upper component has to read the
right endmarker in the final step. In the backward computation
the upper component immediately starts, whereas the lower component
starts with with one time step delay. 
When the upper component has read the rightmost $a$ of the $a$-block
after having passed the $b$-block, the transition functions ensure
that the lower component has to read the leftmost $b$ of the 
$b$-block after having passed the $a$-block. Finally, when the 
upper component has reached the right endmarker, it waits for one
time step. To reach the initial configuration the 
lower component has to read the left endmarker in the next
time step. 

\begin{sloppypar}
For the precise construction of a $\revpwk$ accepting 
the language 
$L=\{\, a^nb^n \mid n \ge 1 \,\}$
we define 
$\mathcal{A}=\langle \{a,b\}, \{a,b,\leftend,\rightend\}, \leftend,\rightend, A_1, A_2 \rangle$ where
$$A_1=\langle \{p_0,p_1, \ldots, p_5\}, \{a,b\}, \delta_1, \mu_1, p_0, \{p_5\} \rangle \mbox{ and }
A_2=\langle \{q_0,q_1, \ldots, q_5\}, \{a,b\}, \delta_2, \mu_2, q_0, \{\} \rangle.$$
The broadcast functions $\mu_1, \mu_2$ and the reverse broadcast functions $\invmu_1,\invmu_2$ are 
defined as \mbox{$\mu_1(p,x)=\invmu_1(p,x) = x$} and $\mu_2(q,x)=\invmu_2(q,x) = x$ 
for all $p \in \{p_0,p_1, \ldots, p_5\}$, $q \in \{q_0,q_1, \ldots, q_5\}$,
and 
\mbox{$x \in \{a,b,\leftend,\rightend\}$.}
The transition functions $\delta_1, \delta_2$ and $\invdelta_1,\invdelta_2$
are as follows.
\end{sloppypar}

\begin{center}
\renewcommand{\arraystretch}{1.1}
\begin{tabular}[t]{rccc}
\hline
\multicolumn{4}{c}{${A_1}$ forward}\\
\hline
(1) &  $\delta_1(p_0,\leftend,\rightend)$ &=& $(p_1, 0)$\\
(2) &  $\delta_1(p_1,\leftend,b)$ &=& $(p_2, +)$\\
(3) &  $\delta_1(p_2,a,b)$ &=& $(p_2, +)$\\
(4) &  $\delta_1(p_2,a,a)$ &=& $(p_3, +)$\\
(5) &  $\delta_1(p_3,b,a)$ &=& $(p_3, +)$\\
(6) &  $\delta_1(p_3,b,\leftend)$ &=& $(p_4, +)$\\
(7) &  $\delta_1(p_4,\rightend,\leftend)$ &=& $(p_5, 0)$\\
\hline
\end{tabular}
\quad
\begin{tabular}[t]{rccc}
\hline
\multicolumn{4}{c}{$A_1$ backward}\\
\hline
(1) &  $\invdelta_1(p_1,\leftend, \rightend)$ &=& $(p_0, 0)$\\
(2) &  $\invdelta_1(p_2,\leftend,b)$ &=& $(p_1, -)$\\
(3) &  $\invdelta_1(p_2,a,b)$ &=& $(p_2, -)$\\
(4) &  $\invdelta_1(p_3,a,a)$ &=& $(p_2, -)$\\
(5) &  $\invdelta_1(p_3,b,a)$ &=& $(p_3, -)$\\
(6) &  $\invdelta_1(p_4,b,\rightend)$ &=& $(p_3, -)$\\
(7) &  $\invdelta_1(p_5,\rightend,\leftend)$ &=& $(p_4, 0)$\\
\hline
\end{tabular}

\medskip
\renewcommand{\arraystretch}{1.1}
\begin{tabular}[t]{rccc}
\hline
\multicolumn{4}{c}{${A_2}$ forward}\\
\hline
(1) &  $\delta_2(q_0,\rightend,\leftend)$ &=& $(q_1, +)$\\
(2) &  $\delta_2(q_1,b,\leftend)$ &=& $(q_2, +)$\\
(3) &  $\delta_2(q_2,b,a)$ &=& $(q_2, +)$\\
(4) &  $\delta_2(q_2,a,a)$ &=& $(q_3, +)$\\
(5) &  $\delta_2(q_3,a,b)$ &=& $(q_3, +)$\\
(6) &  $\delta_2(q_3,\leftend,b)$ &=& $(q_4, 0)$\\
(7) &  $\delta_2(q_4,\leftend,\rightend)$ &=& $(q_5, 0)$\\
\hline
\end{tabular}
\quad
\begin{tabular}[t]{rccc}
\hline
\multicolumn{4}{c}{$A_2$ backward}\\
\hline
(1) &  $\invdelta_2(q_1,\rightend, \leftend)$ &=& $(q_0, -)$\\
(2) &  $\invdelta_2(q_2,b,\leftend)$ &=& $(q_1, -)$\\
(3) &  $\invdelta_2(q_2,b,a)$ &=& $(q_2, -)$\\
(4) &  $\invdelta_2(q_3,a,a)$ &=& $(q_2, -)$\\
(5) &  $\invdelta_2(q_3,a,b)$ &=& $(q_3, -)$\\
(6) &  $\invdelta_2(q_4,\leftend,b)$ &=& $(q_3, 0)$\\
(7) &  $\invdelta_2(q_5,\leftend,\rightend)$ &=& $(q_4, 0)$\\
\hline
\end{tabular}
\end{center}

We note that it is shown in~\cite{kutrib:2012:rpda} that
$L=\{\, a^n b^n \mid n\ge 1\,\}$ is not accepted by any 
reversible pushdown automaton.
\eoe
\end{example}

\begin{example}\label{exa:gladkij}
The non-context-free language 
$L'=\{\, w \dollar w^R \dollar a^{|w|} \mid w \in \{a,b\}^* \,\}$ is 
accepted by a $\revpwk$. Here, the principal idea is that the
upper component waits at the left endmarker, while the lower
component moves across the $a$-block. Having reached the 
second~$\dollar$, both components move with maximum speed and
test the structure $w\dollar w^R$ by communicating in every
time step they read. If no error occurred, the upper component
moves to the second~$\dollar$, while the lower component waits
at the first~$\dollar$. Finally, both components move with 
maximum speed and test the length of $w$ equals the length
of the $a$-block. The moving of the components in the backward
computation is straightforward. 
\eoe
\end{example}

\section{Reversibility versus Irreversibility}\label{sec:rev-versus-irrev}

We now turn to the question of whether reversible two-party Watson-Crick systems
are weaker than irreversible ones or not;
it turns out that they are. In fact, there are languages
accepted by irreversible two-party Watson-Crick systems
that do not need any communication which cannot be accepted 
by any reversible two-party Watson-Crick system
regardless of the number of communications.
To show this claim, we will use regular witness languages. Let
$\Sigma\supseteq \{a,b\}$ be an alphabet and $I\subseteq \Sigma^*$
be regular such that $I=I^R$.
Then we define 
$L_I= \{\, a^{m_1} b v b a^{m_2}\mid m_1,m_2\geq 0, v\in b^* \text{ or }
v\in I\,\}$.
So, the words in $L_I$ have a nonempty prefix of $a$'s, followed
by a $b$, followed by a factor of $b$'s or a factor from $I$, 
followed by a $b$, followed by a nonempty suffix of~$a$'s.

\begin{theorem}\label{theo:reg-not}
Let $\Sigma\supseteq \{a,b\}$ and $I\subseteq \Sigma^*$.
Then language $L_I$ is not accepted by any $\revpwk$.
\end{theorem}

\begin{proof}
Assume for the purpose of contradiction that $L_I$ is accepted by some 
$\revpwk$ $\mathcal{A}$.
Since we do not limit the number of possible communications, we
simply assume that both components send a message at every time
step. In this case, for the sake of easier writing, we can assume that there
is \emph{one common} finite-state control for both components. This control receives
a pair of input symbols in every step, changes the state, and moves the
components if required. Now we can argue that the system is irreversible if
there are two reachable states that have a common successor state for the same
pair of input symbols.

We denote this system $M$, its set of states $Q$, and its transition function $\delta$.
We now consider accepting computations on words $w=a^xb^ya^z \in L_I$, where $x,y,z$ are
long enough. 
In a first phase of such a computation, eventually at least one
component has to start to move across the $a$-prefix or $a$-suffix. Otherwise
the overall computation would loop forever. Since
$L_I=L^R_I$, we can safely assume that the upper
component moves. The lower component may move across the $a$-suffix or stay
stationary on the endmarker or some $a$. We choose $x$ and $z$ large enough
such that $M$ runs into a state cycle in this phase. Moreover, we
choose~$z$ that large that the upper component arrives at the first $b$ after
the $a$-prefix before the lower component has passed the $a$-suffix.
Let $p_1,p_2,\dots, p_k$ be the state cycle. We can adjust the length of the
prefix such that $M$ moves the upper component on the first $b$ while entering
state $p_k$. So, we have a configuration of the 
form 
$p_k\colon \leftend aa\cdots a \ora{b}b\cdots baa\cdots \ula{\sigma}\cdots$,
where the state of $M$ is written in front of $\leftend$, and $\sigma=a$ or 
$\sigma=\rightend$, and the components are scanning the 
symbols indicated by the arrows. Next, we can enlarge $z$ such that $M$ runs again in a 
state loop while the upper component is reading $b$'s and the lower component
is reading $\rightend$ or $a$'s. Assume that the sequence of states passed through is extended
from $p_k$ by $p'_1, p'_2, \dots, p'_i, p''_1,\dots p''_j, p''_1$. Then
we know $\delta(p'_i, (b,\sigma_1))=(p''_1,d_1,d_2)$ and 
$\delta(p''_j, (b,\sigma_2))=(p''_1,d_1,d_2)$, where $d_1,d_2$ indicate
whether the components are moved or not. Since~$M$ is reversible, we derive
$p'_1, p'_2, \dots, p'_i, p''_1,\dots p''_j, p''_1 = p_1,p_2,\dots, p_k,p_1$
or
$(b,\sigma_1)\neq (b,\sigma_2)$ and, thus, $\sigma_1\neq \sigma_2$ and, hence,
$\sigma_1=\rightend$ and $\sigma_2=a$. Dependent on whether the
loop on the $(a,\sigma)$'s is continued on the $(b,\sigma)$'s,
or the second possibility, we distinguish two
cases. A similar distinction will be made in several sub-cases. 

\textbf{Case A} The system $M$ continues to loop through the states
$p_1,p_2,\dots, p_k$ while reading $(b,\sigma)$'s. Recall that 
the current state determines the last movements of the components. Therefore,
the upper component moves across the~$b$'s. Moreover, we can choose $y$ and $z$
again large enough such that the upper component runs through several loops
and $M$ moves the upper component on the first~$a$ of the suffix while entering
state $p_k$. So, we have a configuration of the 
form 
$p_k\colon \leftend aa\cdots a bb\cdots b\ora{a}a \cdots a\ula{\sigma}\cdots$.
Now, we can repeat the argument from above and distinguish the two sub-cases
that $M$ continues to loop through the states $p_1,p_2,\dots, p_k,p_1$,
or 
$(a,\sigma_1)\neq (a,\sigma_2)$ and, thus, $\sigma_1\neq \sigma_2$ and, hence,
$\sigma_1=\rightend$ and $\sigma_2=a$.

\textbf{Case A.1}
The system $M$ continues to loop through the states
$p_1,p_2,\dots, p_k$ while reading $(a,\sigma)$'s. In this sub-case
the upper component may reach the right endmarker before the lower component
reaches the $b$ before the $a$-suffix. Then the remaining computation of $M$ 
is that of a finite automaton, that is, of the lower component. Since the
language $a^*b^*a^*$ is not accepted by any reversible $\dfa$, we obtain a 
contradiction.

Therefore, the upper component may reach the right endmarker not before the 
lower component reaches the $b$ before the $a$-suffix.
Now, again we can repeat the argument from above and distinguish the two sub-cases
that $M$ continues to loop through the states $p_1,p_2,\dots, p_k$
while moving the lower component
or $(a,\sigma_1)$ must not be equal to $(a,\sigma_2)$ which can be violated by
adjusting the value of $z$. In this way $\sigma_1=\sigma_2=a$, a
contradiction. If, however, $M$ continues to loop through the states 
$p_1,p_2,\dots, p_k$, by almost the same arguments as before
we can obtain a contradiction unless $M$ continues to loop through the states 
$p_1,p_2,\dots, p_k$ until the lower component has reached the left
endmarker. In this case, the language $\{a,b\}^+$ is accepted.

\begin{sloppypar}
\textbf{Case A.2}
The sequence of states passed through to reach the configuration
$p_k\colon \leftend aa\cdots a bb\cdots b\ora{a}a \cdots a\ula{\sigma}\cdots$
is extended
from state $p_k$ by the states \mbox{$q'_1, q'_2, \dots, q'_{i'}, q''_1,\dots q''_{j'}, q''_1$,} and
we have
$\delta(q'_{i'}, (a,\sigma_1))=(q''_1,d_1,d_2)$ and 
\mbox{$\delta(q''_{j'}, (a,\sigma_2))=(q''_1,d_1,d_2)$,} and therefore
$(a,\sigma_1)\neq (a,\sigma_2)$ which implies $\sigma_1=\rightend$ and 
$\sigma_2=a$.
\end{sloppypar}

Now, the upper component may or may not reach the right endmarker before the lower component
reaches the $b$ before the $a$-suffix. We obtain a contradiction almost
literally as in Case A.1.

\textbf{Case B}
The sequence of states passed through to reach the configuration
$\leftend aa\cdots a \ora{b}b\cdots baa\cdots \ula{\sigma}\cdots$
in state $p_k$ is extended
from $p_k$ by $p'_1,  \dots, p'_i, p''_1,\dots p''_j,
p''_1$. 
Then we have $\delta(p'_i, (b,\sigma_1))=(p''_1,d_1,d_2)$ and 
$\delta(p''_j, (b,\sigma_2))=(p''_1,d_1,d_2)$, and therefore,
$(b,\sigma_1)\neq (b,\sigma_2)$ which implies 
$\sigma_1=\rightend$ and $\sigma_2=a$.

\textbf{Case B.1}
If the upper component moves in the state cycle $p''_1,\dots p''_j$, 
then we can choose $z$ again large enough such that the upper component
reaches the first $a$ after the $b$-factor before the lower component
reaches the $b$ before the $a$-suffix.
So, a configuration $\leftend aa\cdots a bb\cdots b\ora{a}a \cdots \ula{a}\cdots$
is reached in some state from the cycle. We obtain a contradiction along the
argumentation as in Case A.1.

\textbf{Case B.2}
If the upper component does not move in the state cycle $p''_1,\dots p''_j$, 
then a configuration $\cdots \ora{b} b\cdots \ula{b}aa\cdots$
is reached in some state from the cycle.

Assume that from here the computation continues in the same state cycle
until the lower component has reached the left endmarker. Then the upper
component stays on the current input in this phase, and
the remaining computation of~$M$ is that of a finite automaton, that is, of
the upper component on its remaining input of the form $b^*a^*$, which
is not accepted by any reversible $\dfa$. So, we obtain a contradiction.

We conclude that the computation cannot continue in the same state cycle.
If it continues in some state cycle 
$q'_1, q'_2, \dots, q'_{i'}, q''_1,\dots q''_{j'}, q''_1$ 
while both components read $b$'s, then we have
$\delta(q'_{i'}, (b,b))=(q''_1,d_1,d_2)$ and 
$\delta(p''_{j'}, (b,b))=(p''_1,d_1,d_2)$ which violates the reversibility.

If the computation continues in some state cycle 
$q'_1, q'_2, \dots, q'_{i'}, q''_1,\dots q''_{j'}, q''_1$ 
after at least one component has passed across the $b$-factor, we obtain a similar
contradiction with input pairs $(a,b)$, $(b,a)$, or~$(a,a)$.

This concludes the case analysis. Since in any possible case a contradiction
is derived, the initial assumption that $L_I$ is accepted by some 
$\revpwk$ is wrong and the assertion follows.
\end{proof}

The result of Theorem~\ref{theo:reg-not} that there is a regular
language that is not accepted by any $\revpwk$ together with
Example~\ref{exa:anbn} showing that the non-regular language
$\{\, a^n b^n \mid n\ge 1\,\}$
is accepted by a $\revpwk$ proves that the class of
languages accepted by $\revpwk$ and the regular languages
are incomparable. Since $\{\, a^n b^n \mid n\ge 1\,\}$ is
a linear and real-time deterministic context-free language, 
we immediately obtain the incomparability to the linear context-free
languages as well as to the real-time deterministic context-free
languages. It is shown in~\cite{kutrib:2012:rpda} that every
regular language
can be accepted by a reversible pushdown automaton. 
Moreover, it is shown that the
language $\{\, a^n b^n \mid n\ge 1\,\}$ cannot be accepted by
any reversible pushdown automaton.
Hence,
the classes of languages accepted by $\revpwk$ and reversible
pushdown automata are incomparable as well. 
 
\section{Closure Properties}\label{sec:closures}

The goal of this section is to collect some closure properties
of the families $\lfam(\revpwk)$.
For this purpose, the regular languages $L_I$ can 
be used very well in several cases.
In particular, we consider 
Boolean operations (complementation, union, intersection) and AFL operations 
(union, intersection with regular languages, homomorphism, inverse homomorphism, 
concatenation, iteration). The positive closure under reversal is trivial. 
The results are summarized in Table~\ref{tab:closure} at the end of the section.

\begin{proposition}\label{prop:closure-compl}
The family $\lfam(\revpwk)$ is closed under complementation.
\end{proposition}

\begin{proposition}\label{prop:nonclosure-union}
The family $\lfam(\revpwk)$ is not closed under 
union, intersection, and intersection with regular languages.
\end{proposition}

\begin{proof}
Let $\Sigma=\{a,b\}$. For $I=\emptyset$, we consider the regular language 
$L_\emptyset=\{\, a^{m_1} bb^{m_3}b a^{m_2}\mid m_1,m_2,m_3\geq 0\,\}$.
By Theorem~\ref{theo:reg-not}, the regular language $L_\emptyset$ does not belong to
the family $\lfam(\revpwk)$. On the other hand, the language $\Sigma^*$ does
belong to the family. Since $\Sigma^* \cap L_\emptyset = L_\emptyset$, we obtain the
non-closure under intersection with regular languages.

The non-closure under intersection is witnessed by the languages, 
$L_1= \{\, a^{m} bb v \mid m\geq 0, v\in \{a,b\}^*\,\}$
and
$L_2= \{\, v bb a^{m}\mid m\geq 0, v\in \{a,b\}^*\,\}$.

We show that $L_1$ is accepted by some more or less trivial
$\revpwk$ without any communication as follows.

The lower component does nothing, that is, it loops in its non-accepting initial state
on the right endmarker.
The behavior of the upper component is depicted as a state graph in
Figure~\ref{fig:triv-dfa}.
If and only if the component has seen a correct prefix of the form $a^* bb$ it halts in
an accepting state (the rest of the input cannot affect the computation result
any more and, by definition, there is no need to read it).

\begin{figure}[!ht]
  \centering
  \includegraphics[scale=.9]{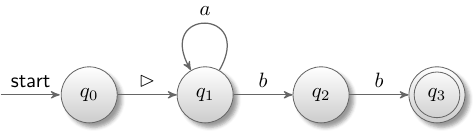} 
  \caption{State graph of the upper component of a $\revpwk$
  accepting $L_1$.}
  \label{fig:triv-dfa}
\end{figure}

Since $L_2=L_1^R$ and the closure of  $\lfam(\revpwk)$ under reversal, we
conclude that $L_2$ belongs to $\lfam(\revpwk)$ as well. However,
$L_1 \cap L_2 = L_I$ for $I= b\{a,b\}^*b$ and, thus, the non-closure under intersection 
follows.

The non-closure under union follows from the closure under complementation
and the non-closure under intersection by De Morgan's law. 
\end{proof}

\begin{proposition}\label{prop:nonclosure-catenation}
The family $\lfam(\revpwk)$ is not closed under 
concatenation and iteration.
\end{proposition}

\begin{proof}
The witness language for both operations is
$L=\{\, a^nb^n\mid n\geq 1\,\}$ which belongs
to $\lfam(\revpwk)$ by Example~\ref{exa:anbn}.

For the concatenation we consider $L\cdot L$ and
for the iteration we consider $L^*$.

Essentially, using a different but similar language, in
\cite{leupold:2010:ftwcsr} it is shown that for $n$ long enough
both components have to scan some symbol from each two factors whose
lengths have to be compared simultaneously. This argument applies also
here. However, the two components can simultaneously stay in two corresponding
factors at most for one such pair. This implies that neither the language $L\cdot L$
nor the language $L^*$ is accepted even by any not necessarily reversible $\dpwk$.
\end{proof}

\begin{proposition}\label{prop:nonclosure-lenpreshom}
The family $\lfam(\revpwk)$ is not closed under 
length-preserving homomorphisms.
\end{proposition}

\begin{proposition}\label{prop:nonclosure-invhom}
The family $\lfam(\revpwk)$ is not closed under 
inverse homomorphisms.
\end{proposition}

\begin{table}[!ht]
\begin{center}
\renewcommand{\arraystretch}{1.2}\setlength{\tabcolsep}{6pt}
\begin{tabular}{|c|c|c|c|c|c|c|c|c|c|}
\hline
  Family & $\overline{\phantom{aa}}$ & $\cup$ & $\cap$ & $\cap_\subtext{reg}$ & $\cdot$ & $*$ &
  $h_{\text{len.pres.}}$ & $h^{-1}$ & $R$\\
\hline\hline
$\revpwk$   & \cyes & \cno & \cno & \cno & \cno & \cno & \cno & \cno & \cyes\\
\hline
\end{tabular}
\end{center}
\caption{Closure properties of the language families discussed.}
\label{tab:closure}
\end{table}

\section{Restricted Communication}\label{sec:restr-comm}

The $\revpwk$ considered in the previous sections may communicate arbitrarily often.
In this section, we want to consider $\dpwk$ and $\revpwk$ with a restricted amount
of communications. According to the definition in Section~\ref{sec:prelim} we have a function
$f : \mathbb{N} \to \mathbb{N}$ and define that a $\dpwk$ is communication bounded
by~$f$ if all words $w$ in the language are accepted with computations where the total
number of messages sent is bounded by~$f(|w|)$. A $\revpwk$ is communication bounded
by~$f$ if, in addition, the total number of messages sent in reverse computations is
bounded by~$f(|w|)$ as well. Here, we will study the language class with constant communication,
where $f \in O(1)$, the class with logarithmic communication, where $f \in O(\log(n))$,
the class with square root communication, where $f \in O(\sqrt{n})$, and the class with
arbitrary, i.e., linear communication, where $f \in O(n)$. The relations of these classes
have been investigated for $\dpwk$ in~\cite{kutrib:2010:tpwcc}. Here, we will complement the results
for $\revpwk$ and clarify the relations between reversible and
general, possibly irreversible, devices.
We start with an example presenting a non-semilinear language that is accepted by a
$\revpwkm{O(\log(n))}$.

\begin{example}\label{exa:a2n}
The language 
$L_{expo}=\{\,a^{2^{0}}b a^{2^{2}} b \cdots b a^{2^{2m}}c a^{2^{2m+1}}b\cdots ba^{2^{3}}ba^{2^{1}}
\mid m \ge 1\,\}$ 
is accepted by a $\revpwk$.
The rough idea of the construction is that in a first phase
the components compare the lengths ${2^{0}}$ with ${2^{1}}$,
${2^{2}}$ with ${2^{3}}$, \dots, and $2^{2m}$ with $2^{2m+1}$. 
The first phase ends when both components reach the center symbol $c$. 
In a second phase, the components compare the length $2^{2m}$ with $2^{2m-1}$,
$2^{2m-2}$ with $2^{2m-3}$, \dots, and ${2^{2}}$ with ${2^{1}}$.
To achieve this the lower component has to wait on the $c$ until
the upper component has moved across the block $a^{2^{2m+1}}$.
To realize the comparisons, the upper component moves across its $a$-blocks 
with half speed, whereas the lower component moves across its $a$-blocks with 
full speed, that is, one square per step.
The length comparisons in the first and second phase are checked
by communicating when a $b$, $c$, or the right endmarker is reached which must happen
synchronously. 

The length of an accepted input is $n=2^{2m+2}+2m$. There are communications
only on symbols $b$, $c$, and $\rightend$ both in forward computations and
reverse computations. Hence, there are at most $2m+3$ communications
in forward computations as well as in reverse computations. Thus,
the $\revpwk$ constructed is a $\revpwkm{O(\log(n))}$ and $L_{expo}$ belongs to $\lfam(\revpwkm{O(\log(n))})$.
\eoe
\end{example}

\begin{lemma}\label{lem:linear}
The language 
${L}_{lin}=\{\, wcw^R \mid w\in \{0,1\}^*\,\}$ 
belongs to 
$\lfam(\revpwkm{O(n)})$.
\end{lemma}

\begin{proof}
A $\revpwk$ accepting ${L}_{lin}$
will move its both components synchronously towards the middle marker~$c$
as long as the input symbol read and communicated in every step
is equal. In case of inequivalence the computation halts 
non-accepting. If both components reach the middle marker~$c$ at the
same time, the first task is nearly accomplished. It remains for
the lower component to read the input completely and to halt non-accepting in case of another
symbol~$c$ occurring. Since both components move synchronously and
communicate in every step, it is clear that ${L}_{lin}$ can be
accepted by a $\revpwkm{O(n)}$. 
\end{proof}

As a combination of Example~\ref{exa:a2n} and Lemma~\ref{lem:linear} we obtain the following lemma.

\begin{lemma}\label{lem:expolinear}
$
\hat{L}_{expo} = 
\{\,a^{2^{0}}x_1 a^{2^{2}} x_2 \cdots x_{m} a^{2^{2m}}c a^{2^{2m+1}}x_m \cdots x_2
 a^{2^{3}} x_1 a^{2^{1}}
\mid m \ge 1 \mbox{ and } x_i \in \{0,1\}, 1\leq i\leq m\,\}
$
belongs to 
$\lfam(\revpwkm{O(\log(n))})$.
\end{lemma}

\begin{proof}
It can be observed from the construction in Example~\ref{exa:a2n} that in the
first phase both components communicate on every symbol $b$ and~$c$. So, on
the corresponding input from $\hat{L}_{expo}$ both
components can communicate on every symbol 
$0$, $1$, and $c$ in order to simulate the $\revpwk$ accepting ${L}_{lin}$ as a subtask.
\end{proof}

With similar ideas it is possible to show the following lemma.

\begin{lemma}\label{lem:quadlinear}
$
\hat{L}_{poly} = 
\{\,ax_1 a^{5} x_2 \cdots x_m a^{4m+1}c a^{4m+3}x_m \cdots x_2 a^{7} x_1 a^{3}
\mid \mbox{$m \ge 0$}
\mbox{ and } x_i \in \{0,1\}, 1\leq i\leq m\,\}
$
belongs to 
$\lfam(\revpwkm{O(\sqrt{n})})$.
\end{lemma}

It is shown in~\cite{kutrib:2010:tpwcc} that ${L}_{lin}$ does
not belong to $\lfam(\dpwkm{O(f)})$ if $f \in \frac{n}{\omega(\log(n))}$. Hence,
${L}_{lin}$ does not belong to $\lfam(\revpwkm{O(\sqrt{n})})$.
It is also shown in~\cite{kutrib:2010:tpwcc} that $\hat{L}_{poly}$ does not
belong to $\lfam(\dpwkm{O(f)})$ if $f \in O(\log(n))$. Thus, $\hat{L}_{poly}$ does not
belong to $\lfam(\revpwkm{O(\log(n))})$. Finally, it is known due to~\cite{kutrib:2010:tpwcc}
that every language in $\lfam(\dpwkm{O(1)})$ is semilinear. Since $\hat{L}_{expo}$ is not
semilinear, it does not belong to $\lfam(\revpwkm{O(1)})$.
Together with Lemma~\ref{lem:linear}, Lemma~\ref{lem:expolinear}, and Lemma~\ref{lem:quadlinear} 
we obtain the following proper hierarchy:
\begin{multline*}
\lfam(\revpwkm{O(1)}) \subset \lfam(\revpwkm{O(\log(n))})
 \subset\\
\lfam(\revpwkm{O(\sqrt{n})}) \subset \lfam(\revpwkm{O(n)})
\end{multline*}

Theorem~\ref{theo:reg-not} presents a regular language that is not accepted by any
$\revpwkm{O(n)}$. Since the regular languages belong to $\lfam(\dpwkm{O(1)})$ we
immediately obtain proper inclusions between reversible and general language classes
with the same amount of communication. These results and the other results of this
section are summarized in Figure~\ref{fig:pic-incl}.

\begin{figure}[!ht]
  \centering
  \includegraphics[scale=.7]{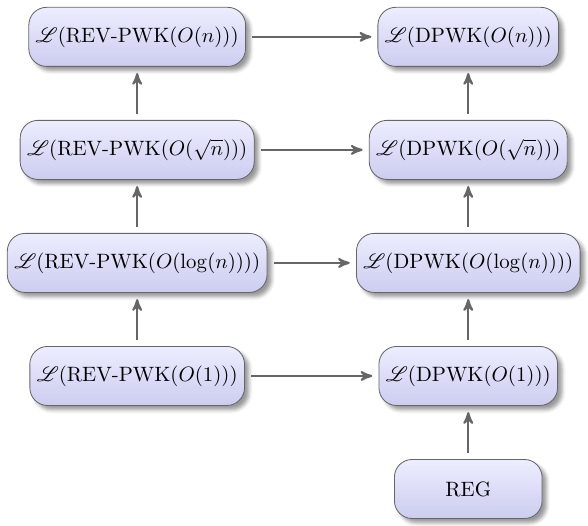} 
  \caption{Relationships between language families induced by
two-party Watson-Crick systems. An arrow between families indicates a strict
inclusion.
}\label{fig:pic-incl}
\end{figure}

\section{Decidability Questions}\label{sec:deci}

In this section, we will discuss several decidability
questions for $\revpwk$. It has been shown in~\cite{kutrib:2010:tpwcc}
that the questions of emptiness, finiteness, inclusion, and
equivalence are decidable for general, possibly irreversible,
$\dpwk$ in case of a finite number of communications. 
This result leads immediately to the following decidability
results for $\revpwk$ in case of a finite number of communications.

\begin{theorem}\label{theo:dec}
Let $k\geq 0$ be a constant. Then 
emptiness, finiteness, inclusion, and equivalence
are decidable for $\revpwkm{k}$.
\end{theorem}

Next, we want to obtain that the decidability questions become
undecidable if a non-constant number of communications is
used. In a first step, we show that the questions of
emptiness, finiteness, inclusion, and equivalence are 
undecidable and, moreover, not even
semidecidable for $\revpwk$ in case of a linear number of
communications used. In a second step, we will obtain the same
non-semidecidability results with a superlogarithmic number
of communications used.

It has been shown in~\cite{kutrib:2017:owrmhfa} that the questions of
testing emptiness, finiteness, inclusion, and equivalence
are not semidecidable for
\emph{reversible two-head finite automata}. The difference between
such automata and $\dpwk$ is that the former move their two heads
in the same direction from left to right, whereas the latter move
both heads in opposite directions. Now, the idea is to simulate
a reversible two-head finite automaton by a $\revpwk$.

The non-semidecidability results for reversible two-head finite
automata are obtained by showing that the set $\valc_M$ of 
suitably encoded valid computations of a deterministic linearly
space bounded one-tape, one-head Turing machine~$M$, so-called
linear bounded automaton ($\lba$) can be accepted by a 
reversible two-head finite automaton. It should be noted that
the due to technical reasons the definition of the set 
$\valc_M$ in~\cite{kutrib:2017:owrmhfa}
considers valid computations on inputs of length at least~2.

Now, we will construct a $\revpwkm{O(n)}$ that accepts
the set $\valc'_M = \{\, w^R c w \mid w \in \valc_M \,\}$, 
where the set $\valc_M$ is defined over some alphabet $A$ and
$c \not\in A$ is a new symbol.

\begin{lemma}\label{lem:valc'}
Let $M$ be an $\lba$. Then, a $\revpwkm{O(n)}$ accepting $\valc'_M$
can effectively be constructed.
\end{lemma}

\begin{proof}
Let $M$ be an $\lba$. A $\revpwk$~$M'$ accepting $\valc'_M$ has
to accomplish two tasks. 
First, $M'$ will test the structure $w^Rcw$ disregarding whether
$w$ belongs to $\valc_M$ or not. To achieve this task 
we use a similar approach as described in the proof of Lemma~\ref{lem:linear}.
Both 
components will move synchronously towards the middle marker~$c$
as long as the input symbol read and communicated in every step
is equal. The structure $w^Rcw$ is correctly tested, if both components 
reach the middle marker~$c$ at the same time. Then, the first task is 
nearly accomplished, but it remains for the lower component, while 
accomplishing the second task, to read the input completely and to 
halt non-accepting in case of another symbol~$c$ occurring. 
Since both components move synchronously and
communicate in every step, it is clear that the first task can be
realized by a $\revpwkm{O(n)}$. 

For the second task, we first observe that the remaining input
for both components is the same word~$w$ and it remains to be 
checked whether or not~$w$ belongs to $\valc_M$. This can now be
realized by implementing the construction given in~\cite{kutrib:2017:owrmhfa}
for two-head finite automata. The head~1 is simulated by the upper
component and head~2 is simulated by the lower component, whereby
the middle marker~$c$ is interpreted as the left endmarker for the
two-head finite automaton. In this construction the lower
component reads the input completely and can halt non-accepting 
if another symbol~$c$ is read. Since the two-head finite automaton
is reversible, the second task and, therefore, the complete 
construction can be realized by a $\revpwkm{O(n)}$. 
\end{proof}

This leads immediately to the following non-semidecidability results.

\begin{theorem}\label{theo:undec-base}
\quad 
The problems of testing
emptiness, finiteness, inclusion, and equivalence are not
semidecidable for a given $\revpwkm{O(n)}$.
\end{theorem}

\begin{proof}
Let $M$ be an $\lba$ accepting inputs over the alphabet $\Sigma$. 
According to Lemma~\ref{lem:valc'} we can effectively construct
a $\revpwkm{O(n)}$ $M'$ accepting $\valc'_M$. 
Clearly, $L(M')=\valc'_M$ is empty if and only if
$\valc_M$ is empty if and only if
$L(M)$ is either empty or contains some
words from the finite set $\{\lambda\} \cup \Sigma$.
The latter words have to be considered, since $M$ may accept
words of length less than two.
Since the word problem is decidable for $\lba$s and 
emptiness is not semidecidable for $\lba$s, the 
non-semidecidability of emptiness follows.  

We also obtain that $L(M')=\valc'_M$ is finite if and only if
$\valc_M$ is finite if and only if $L(M)$ is finite.
Since finiteness is not semidecidable for $\lba$s, the 
non-semidecidability of finiteness follows.

Finally, it is easy to effectively construct a $\revpwkm{1}$ 
that accepts nothing. Hence, the non-semidecidability
of equivalence and inclusion follows immediately.
\end{proof}

Our next step is to obtain these non-semidecidability results also
for $\revpwk$ with less communication. Our approach is 
to define another variant of $\valc'_M$ in which
the length of each configuration is enlarged while 
the same amount of communication is being kept. 
A similar approach has been used in~\cite{kutrib:2010:tpwcc} for general,
possibly irreversible, $\dpwk$. However, here the details 
are quite different and more complicated since the construction has
to be reversible. The detailed and lengthy construction is omitted here. 
With all these prerequisites it is possible to show the following theorem.

\begin{theorem}\label{theo:undec-log}
\quad 
The problems of testing
emptiness, finiteness, inclusion, and equivalence
are not semidecidable for a given 
$\revpwkm{O(\log(n)\cdot \log(\log(n)))}$.
\end{theorem}


\begin{thebibliography}{10}
\providecommand{\bibitemdeclare}[2]{}
\providecommand{\surnamestart}{}
\providecommand{\surnameend}{}
\providecommand{\urlprefix}{Available at }
\providecommand{\url}[1]{\texttt{#1}}
\providecommand{\href}[2]{\texttt{#2}}
\providecommand{\urlalt}[2]{\href{#1}{#2}}
\providecommand{\doi}[1]{doi:\urlalt{https://doi.org/#1}{#1}}
\providecommand{\eprint}[1]{arXiv:\urlalt{https://arxiv.org/abs/#1}{#1}}
\providecommand{\bibinfo}[2]{#2}

\bibitemdeclare{article}{Angluin:1982:irl}
\bibitem{Angluin:1982:irl}
\bibinfo{author}{Dana \surnamestart Angluin\surnameend} (\bibinfo{year}{1982}):
  \emph{\bibinfo{title}{Inference of reversible languages}}.
\newblock {\slshape \bibinfo{journal}{J. ACM}}
  \bibinfo{volume}{29}(\bibinfo{number}{3}), pp. \bibinfo{pages}{741--765},
\doi{10.1145/322326.322334}.

\bibitemdeclare{article}{Bennet:1973:lrc}
\bibitem{Bennet:1973:lrc}
\bibinfo{author}{Charles~H. \surnamestart Bennett\surnameend}
  (\bibinfo{year}{1973}): \emph{\bibinfo{title}{Logical Reversibility of
  Computation}}.
\newblock {\slshape \bibinfo{journal}{IBM J. Res. Dev.}} \bibinfo{volume}{17},
  pp. \bibinfo{pages}{525--532}, \doi{10.1147/rd.176.0525}.

\bibitemdeclare{article}{bordihn:2021:rpcfas}
\bibitem{bordihn:2021:rpcfas}
\bibinfo{author}{Henning \surnamestart Bordihn\surnameend} \&
  \bibinfo{author}{Gy{\"{o}}rgy \surnamestart Vaszil\surnameend}
  (\bibinfo{year}{2021}): \emph{\bibinfo{title}{Reversible parallel communicating finite automata systems}}.
\newblock {\slshape \bibinfo{journal}{Acta Inf.}}
  \bibinfo{volume}{58}(\bibinfo{number}{4}), pp. \bibinfo{pages}{263--279},
  \doi{10.1007/s00236-021-00396-9}.

\bibitemdeclare{article}{chatterjee:2017:rwka}
\bibitem{chatterjee:2017:rwka}
\bibinfo{author}{Kingshuk \surnamestart Chatterjee\surnameend} \&
  \bibinfo{author}{Kumar~Sankar \surnamestart Ray\surnameend}
  (\bibinfo{year}{2017}): \emph{\bibinfo{title}{Reversible Watson-Crick
  automata}}.
\newblock {\slshape \bibinfo{journal}{Acta Inf.}}
  \bibinfo{volume}{54}(\bibinfo{number}{5}), pp. \bibinfo{pages}{487--499},
  \doi{10.1007/s00236-016-0267-0}.

\bibitemdeclare{article}{chatterjee:2017:wkpa}
\bibitem{chatterjee:2017:wkpa}
\bibinfo{author}{Kingshuk \surnamestart Chatterjee\surnameend} \&
  \bibinfo{author}{Kumar~Sankar \surnamestart Ray\surnameend}
  (\bibinfo{year}{2017}): \emph{\bibinfo{title}{Watson-Crick pushdown
  automata}}.
\newblock {\slshape \bibinfo{journal}{Kybernetika}}
  \bibinfo{volume}{53}(\bibinfo{number}{5}), pp. \bibinfo{pages}{868--876},
  \doi{10.14736/kyb-2017-5-0868}.

\bibitemdeclare{article}{czeizler:2009:odcwca}
\bibitem{czeizler:2009:odcwca}
\bibinfo{author}{Elena \surnamestart Czeizler\surnameend},
  \bibinfo{author}{Eugen \surnamestart Czeizler\surnameend},
  \bibinfo{author}{Lila \surnamestart Kari\surnameend} \& \bibinfo{author}{Kai
  \surnamestart Salomaa\surnameend} (\bibinfo{year}{2009}):
  \emph{\bibinfo{title}{On the descriptional complexity of {W}atson-{C}rick
  automata}}.
\newblock {\slshape \bibinfo{journal}{Theor. Comput. Sci.}}
  \bibinfo{volume}{410}, pp. \bibinfo{pages}{3250--3260}, \doi{10.1016/j.tcs.2009.05.001}.

\bibitemdeclare{inproceedings}{Freund:1997:wcfa:proc}
\bibitem{Freund:1997:wcfa:proc}
\bibinfo{author}{Rudolf \surnamestart Freund\surnameend},
  \bibinfo{author}{Gheorghe \surnamestart P{\u{a}}un\surnameend},
  \bibinfo{author}{Grzegorz \surnamestart Rozenberg\surnameend} \&
  \bibinfo{author}{Arto \surnamestart Salomaa\surnameend}
  (\bibinfo{year}{1997}): \emph{\bibinfo{title}{{W}atson-{C}rick Finite
  Automata}}.
\newblock In: {\slshape \bibinfo{booktitle}{{DIMACS} Workshop on {DNA} Based
  Computers}}, \bibinfo{publisher}{University of Pennsylvania},
  \bibinfo{address}{Philadelphia}, pp. \bibinfo{pages}{305--317}, \doi{10.1090/dimacs/048/22}.

\bibitemdeclare{article}{holzer:2018:mrdfa}
\bibitem{holzer:2018:mrdfa}
\bibinfo{author}{Markus \surnamestart Holzer\surnameend},
  \bibinfo{author}{Sebastian \surnamestart Jakobi\surnameend} \&
  \bibinfo{author}{Martin \surnamestart Kutrib\surnameend}
  (\bibinfo{year}{2018}): \emph{\bibinfo{title}{Minimal Reversible
  Deterministic Finite Automata}}.
\newblock {\slshape \bibinfo{journal}{Int. J. Found. Comput. Sci.}}
  \bibinfo{volume}{29}, pp. \bibinfo{pages}{251--270}, \doi{10.1142/S0129054118400063}.

\bibitemdeclare{article}{kocman:2022:ajwkfam}
\bibitem{kocman:2022:ajwkfam}
\bibinfo{author}{Radim \surnamestart Kocman\surnameend},
  \bibinfo{author}{Zbynek \surnamestart Krivka\surnameend},
  \bibinfo{author}{Alexander \surnamestart Meduna\surnameend} \&
  \bibinfo{author}{Benedek \surnamestart Nagy\surnameend}
  (\bibinfo{year}{2022}): \emph{\bibinfo{title}{A jumping 5{\({'}\)}
  {\(\rightarrow\)} 3{\({'}\)} Watson-Crick finite automata model}}.
\newblock {\slshape \bibinfo{journal}{Acta Inf.}}
  \bibinfo{volume}{59}(\bibinfo{number}{5}), pp. \bibinfo{pages}{557--584},
  \doi{10.1007/s00236-021-00413-x}.

\bibitemdeclare{inproceedings}{kondacs:1997:pqfsa}
\bibitem{kondacs:1997:pqfsa}
\bibinfo{author}{Attila \surnamestart Kondacs\surnameend} \&
  \bibinfo{author}{John \surnamestart Watrous\surnameend}
  (\bibinfo{year}{1997}): \emph{\bibinfo{title}{On the Power of Quantum Finite
  State Automata}}.
\newblock In: {\slshape \bibinfo{booktitle}{Foundations of Computer Science
  (FOCS 1997)}}, \bibinfo{publisher}{IEEE Computer Society}, pp.
  \bibinfo{pages}{66--75}, \doi{10.1109/SFCS.1997.646094}.

\bibitemdeclare{inproceedings}{kutrib:2014:arca}
\bibitem{kutrib:2014:arca}
\bibinfo{author}{Martin \surnamestart Kutrib\surnameend}
  (\bibinfo{year}{2014}): \emph{\bibinfo{title}{Aspects of Reversibility for
  Classical Automata}}.
\newblock In \bibinfo{editor}{C.~S. \surnamestart Calude\surnameend},
  \bibinfo{editor}{G.~R. \surnamestart Freivalds\surnameend} \&
  \bibinfo{editor}{K.~\surnamestart Iwama\surnameend}, editors: {\slshape
  \bibinfo{booktitle}{Computing with New Resources}}, {\slshape
  \bibinfo{series}{LNCS}} \bibinfo{volume}{8808},
  \bibinfo{publisher}{Springer}, pp. \bibinfo{pages}{83--98}, \doi{10.1007/978-3-319-13350-8\_7}.

\bibitemdeclare{inproceedings}{kutrib:2010:tpwcc}
\bibitem{kutrib:2010:tpwcc}
\bibinfo{author}{Martin \surnamestart Kutrib\surnameend} \&
  \bibinfo{author}{Andreas \surnamestart Malcher\surnameend}
  (\bibinfo{year}{2011}): \emph{\bibinfo{title}{Two-Party {W}atson-{C}rick
  Computations}}.
\newblock In: {\slshape \bibinfo{booktitle}{Implementation and Application of
  Automata (CIAA 2010)}}, {\slshape \bibinfo{series}{LNCS}}
  \bibinfo{volume}{6482}, \bibinfo{publisher}{Springer}, pp.
  \bibinfo{pages}{191--200}, \doi{10.1007/978-3-642-18098-9\_21}.

\bibitemdeclare{article}{kutrib:2012:rpda}
\bibitem{kutrib:2012:rpda}
\bibinfo{author}{Martin \surnamestart Kutrib\surnameend} \&
  \bibinfo{author}{Andreas \surnamestart Malcher\surnameend}
  (\bibinfo{year}{2012}): \emph{\bibinfo{title}{Reversible Pushdown Automata}}.
\newblock {\slshape \bibinfo{journal}{J. Comput. Syst. Sci.}}
  \bibinfo{volume}{78}, pp. \bibinfo{pages}{1814--1827}, \doi{10.1016/j.jcss.2011.12.004}.

\bibitemdeclare{article}{kutrib:2017:owrmhfa}
\bibitem{kutrib:2017:owrmhfa}
\bibinfo{author}{Martin \surnamestart Kutrib\surnameend} \&
  \bibinfo{author}{Andreas \surnamestart Malcher\surnameend}
  (\bibinfo{year}{2017}): \emph{\bibinfo{title}{One-way reversible multi-head
  finite automata}}.
\newblock {\slshape \bibinfo{journal}{Theor. Comput. Sci.}}
  \bibinfo{volume}{682}, pp. \bibinfo{pages}{149--164},
  \doi{10.1016/j.tcs.2016.11.006}.

\bibitemdeclare{inproceedings}{kutrib:2022:rcoowca}
\bibitem{kutrib:2022:rcoowca}
\bibinfo{author}{Martin \surnamestart Kutrib\surnameend} \&
  \bibinfo{author}{Andreas \surnamestart Malcher\surnameend}
  (\bibinfo{year}{2022}): \emph{\bibinfo{title}{Reversible Computations of
  One-Way Counter Automata}}.
\newblock In \bibinfo{editor}{Henning \surnamestart Bordihn\surnameend},
  \bibinfo{editor}{G{\'{e}}za \surnamestart Horv{\'{a}}th\surnameend} \&
  \bibinfo{editor}{Gy{\"{o}}rgy \surnamestart Vaszil\surnameend}, editors:
  {\slshape \bibinfo{booktitle}{{NCMA} 2022}}, {\slshape
  \bibinfo{series}{{EPTCS}}} \bibinfo{volume}{367}, pp.
  \bibinfo{pages}{126--142}, \doi{10.4204/EPTCS.367.9}.

\bibitemdeclare{article}{kutrib:2016:rqa}
\bibitem{kutrib:2016:rqa}
\bibinfo{author}{Martin \surnamestart Kutrib\surnameend},
  \bibinfo{author}{Andreas \surnamestart Malcher\surnameend} \&
  \bibinfo{author}{Matthias \surnamestart Wendlandt\surnameend}
  (\bibinfo{year}{2016}): \emph{\bibinfo{title}{Reversible Queue Automata}}.
\newblock {\slshape \bibinfo{journal}{Fund. Inform.}} \bibinfo{volume}{148},
  pp. \bibinfo{pages}{341--368}, \doi{10.3233/FI-2016-1438}.

\bibitemdeclare{article}{lecerf:lmmtr:1963}
\bibitem{lecerf:lmmtr:1963}
\bibinfo{author}{Yves \surnamestart Lecerf\surnameend} (\bibinfo{year}{1963}):
  \emph{\bibinfo{title}{Logique Math{\'e}matique: {M}achines de {T}uring
  r{\'e}versible}}.
\newblock {\slshape \bibinfo{journal}{C. R. S{\'e}ances Acad. Sci.}}
  \bibinfo{volume}{257}, pp. \bibinfo{pages}{2597--2600}.

\bibitemdeclare{article}{leupold:2010:ftwcsr}
\bibitem{leupold:2010:ftwcsr}
\bibinfo{author}{Peter \surnamestart Leupold\surnameend} \&
  \bibinfo{author}{Benedek \surnamestart Nagy\surnameend}
  (\bibinfo{year}{2010}): \emph{\bibinfo{title}{{$5'\to 3'$} {W}atson-{C}rick
  Automata with Several Runs}}.
\newblock {\slshape \bibinfo{journal}{Fund. Inform.}} \bibinfo{volume}{104},
  pp. \bibinfo{pages}{71--91}, \doi{10.3233/FI-2010-336}.

\bibitemdeclare{article}{morita:2011:twrmhfa}
\bibitem{morita:2011:twrmhfa}
\bibinfo{author}{Kenichi \surnamestart Morita\surnameend}
  (\bibinfo{year}{2011}): \emph{\bibinfo{title}{Two-Way Reversible Multi-Head
  Finite Automata}}.
\newblock {\slshape \bibinfo{journal}{Fund. Inform.}} \bibinfo{volume}{110},
  pp. \bibinfo{pages}{241--254}, \doi{10.3233/FI-2011-541}.

\bibitemdeclare{inproceedings}{nagy:2007:ofpt3pswcfa:proc}
\bibitem{nagy:2007:ofpt3pswcfa:proc}
\bibinfo{author}{Benedek \surnamestart Nagy\surnameend} (\bibinfo{year}{2007}):
  \emph{\bibinfo{title}{On {$5'\to 3'$} Sensing {W}atson-{C}rick Finite
  Automata}}.
\newblock In: {\slshape \bibinfo{booktitle}{{DNA} Computing}}, {\slshape
  \bibinfo{series}{LNCS}} \bibinfo{volume}{4848},
  \bibinfo{publisher}{Springer}, pp. \bibinfo{pages}{256--262}, \doi{10.1007/978-3-540-77962-9\_27}.

\bibitemdeclare{article}{nagy:2013:h5s3sswcfal}
\bibitem{nagy:2013:h5s3sswcfal}
\bibinfo{author}{Benedek \surnamestart Nagy\surnameend} (\bibinfo{year}{2013}):
  \emph{\bibinfo{title}{On a hierarchy of 5{\({'}\)} {\(\rightarrow\)}
  3{\({'}\)} sensing {W}atson-{C}rick finite automata languages}}.
\newblock {\slshape \bibinfo{journal}{J. Log. Comput.}} \bibinfo{volume}{23},
  pp. \bibinfo{pages}{855--872}, \doi{10.1093/logcom/exr049}.

\bibitemdeclare{article}{nagy:2020:wkpa}
\bibitem{nagy:2020:wkpa}
\bibinfo{author}{Benedek \surnamestart Nagy\surnameend} (\bibinfo{year}{2020}):
  \emph{\bibinfo{title}{5{\({'}\)}{\(\rightarrow\)}3{\({'}\)} Watson-Crick
  pushdown automata}}.
\newblock {\slshape \bibinfo{journal}{Inf. Sci.}} \bibinfo{volume}{537}, pp.
  \bibinfo{pages}{452--466}, \doi{10.1016/j.ins.2020.06.031}.

\bibitemdeclare{article}{nagy:2021:od1lswkfst}
\bibitem{nagy:2021:od1lswkfst}
\bibinfo{author}{Benedek \surnamestart Nagy\surnameend} \&
  \bibinfo{author}{Zita \surnamestart Kov{\'{a}}cs\surnameend}
  (\bibinfo{year}{2021}): \emph{\bibinfo{title}{On deterministic 1-limited
  sensing 5{\({'}\)} {\(\rightarrow\)} 3{\({'}\)} Watson-Crick finite-state
  transducers}}.
\newblock {\slshape \bibinfo{journal}{{RAIRO} Theor. Informatics Appl.}}
  \bibinfo{volume}{55}, p.~\bibinfo{pages}{5}, \doi{10.1051/ita/2021007}.

\bibitemdeclare{inproceedings}{nagy:2017:answkac}
\bibitem{nagy:2017:answkac}
\bibinfo{author}{Benedek \surnamestart Nagy\surnameend},
  \bibinfo{author}{Shaghayegh \surnamestart Parchami\surnameend} \&
  \bibinfo{author}{Hamid Mir~Mohammad \surnamestart Sadeghi\surnameend}
  (\bibinfo{year}{2017}): \emph{\bibinfo{title}{A New Sensing 5{\({'}\)}
  {\(\rightarrow\)} 3{\({'}\)} Watson-Crick Automata Concept}}.
\newblock In \bibinfo{editor}{Erzs{\'{e}}bet \surnamestart
  Csuhaj{-}Varj{\'{u}}\surnameend}, \bibinfo{editor}{P{\'{a}}l \surnamestart
  D{\"{o}}m{\"{o}}si\surnameend} \& \bibinfo{editor}{Gy{\"{o}}rgy \surnamestart
  Vaszil\surnameend}, editors: {\slshape \bibinfo{booktitle}{{AFL} 2017}},
  {\slshape \bibinfo{series}{{EPTCS}}} \bibinfo{volume}{252}, pp.
  \bibinfo{pages}{195--204}, \doi{10.4204/EPTCS.252.19}.

\end{thebibliography}

\end{document}